\documentclass[aps,nofootinbib,preprintnumbers,citeautoscript]{revtex4}
\usepackage{amsmath,amssymb}
\usepackage{enumerate}
\usepackage{graphicx}
\usepackage{amsthm}
 \usepackage{mathtools}
 \usepackage{textcomp}
  \usepackage{braket}

\date{\today}

\begin{document}
\title{On the Power of Quantum Queue Automata in Real-time}

\author{Amandeep Singh Bhatia$^\ast$ and  Ajay Kumar$^\dagger$  \\
\textit{Department of Computer Science, Thapar university, India} \\
E-mail: $^\ast$amandeepbhatia.singh@gmail.com}

\begin{abstract}
This paper proposed a quantum analogue of classical queue automata by using the definition of the quantum Turing machine and quantum finite-state automata. However, quantum automata equipped with storage medium of a stack has been considered, but the concept of quantum queue automata has not been introduced so far. The classical Turing machines can be simulated by classical queue automata. Motivated by the efficiency of the quantum Turing machine and nature of classical queue automata, we have introduced the notion of quantum queue automata using unitary criteria. Our contributions are as follows. We have also introduced a generalization of real-time deterministic queue automata, the real-time quantum queue automata which work in real-time i.e. the input head can move towards the right direction only and takes exactly one step per input symbol. We have shown that real-time quantum queue automata is more superior than its real-time classical variants by using quantum transitions. We have proved the existence of the language that can be recognized by real-time quantum queue automata and cannot be recognized by real-time deterministic (reversible) queue automata. Further, we have shown that there is a language that can be recognized by real-time quantum queue automata but not by real-time non-deterministic queue automata. 
\end{abstract}
\maketitle



\theoremstyle{plain}

\newtheorem{thm}{Theorem}

\theoremstyle{definition}
\newtheorem{defn}{Definition}
\newtheorem{exmp}{Example}

\section{Introduction and motivation}
Quantum computing combines quantum physics, computer science, and mathematics by studying computational models based on quantum physics. Quantum computing is based on the quantum phenomena of entanglement and superposition to perform operations on quantum computers (which is substantially different from classical computers) \cite{1}. It may allow us to perform computational tasks which are neither possible nor efficient. In 1994, Shor \cite{2} designed a quantum algorithm for calculating the factor of a large number $\it n$ with space complexity $O(log ~n)$ and runs in $O((log ~n)^2 * log log ~n)$ on a quantum computer, and then perform $O(log ~n)$ post processing time on a classical computer, which could be applied for cracking the RSA algorithm at Bell Laboratories (US). In 1996, Grover \cite{3} designed an algorithm for searching an element in an unsorted database set of size \textit{n} in $\sqrt{n}$ operations approximately.

Soon after the design of Shor's algorithm, interest in quantum computation and information has been increased significantly. A number of different models, their language recognition capability and properties have been introduced \cite{6}. Till now, quantum versions for various classical automata's has been introduced such as quantum Turing machine (QTM) \cite{8} of Turing machine; quantum finite automata (Moore and Crutchfield \cite{4}; Kondacs and Watrous \cite{5}) of deterministic finite automata; quantum pushdown automata (Moore and Crutchfield \cite{4}; Golovkins \cite{9}; Gudder \cite{10}; Qiu \cite{11}) of pushdown automata, quantum real-time one-counter automaton (rtQ1CA) Say et al. \cite{28} of classical real-time one-counter automaton (rt1CA) and many more since last two decades. Some of these constructions are more powerful than their classical counterparts \cite{32}.

In 1936, the same fruitful year Turing introduced the Turing machine, Post \cite{12, 13, 15} proposed the concept of Post machine (PM), which is equivalent to Turing machine model in computation. Later on, Manna \cite{14} followed the approach of Post and defined a automaton equipped with data structure queue called Post machine. Till now, it has been extensively studied from several perspectives. Vollmar \cite{16} investigated the automaton employed with buffer storage, which is basically queue automaton and proved that it can be designed for recursive enumerable languages. In 1980, Brandenburg \cite{24} considered queue automaton (called post machine) and investigated its features. It has also characterized the class of languages recognized by multi-reset machines and equality sets. The computation power of queue automata by putting restriction on time has been investigated on several occasions. Cherubini et al. \cite{29} considered the queue automaton which works in quasi real-time i.e. limit the number of $\epsilon$-transitions by a constant and proved that emptiness problem is undecidable. In 2013, Jakobi et al. \cite{30} introduced queue automaton of constant length and studied its descriptional complexity. Recently, 
Kutrib et al. \cite{25, 26} introduced the concept of reversible queue automata. It has been shown that the class of languages recognized by reversible queue automata strictly consists regular languages. It has been investigated that the computational power of reversible queue automata and Turing machines is an equivalent. Further, the working of reversible queue automata in quasi real-time has been shown. Moreover, the language recognition power of reversible queue automata is compared with reversible pushdown automata and input-driven queue automata and examined their closure properties.

Scegulnaja \cite{31} introduced the concept of real-time quantum Turing machine and proved that a language $L=\{wcxcw^Rcx^R \mid w, x \in \{0, 1\}^*\}$ can be recognized by real-time quantum Turing machine with single work-tape, but cannot be recognized by real-time deterministic Turing machine with single work-tape. In this paper, we have introduced a notion of quantum queue automata. Further, we have shown that a language which cannot be recognized by any real-time deterministic queue automata can be recognized by quantum queue automata in real-time. Moreover, we have proved that real-time quantum queue automata is more powerful than real-time non-deterministic queue automata in language recognition. The organization of the paper is as follows: In section 2, some preliminaries and definitions are given. Further, we have introduced the concept of quantum queue automata with its well-formedness conditions in Section 3. Furthermore, the power of real-time quantum queue automata over classical counterparts is shown in Section 4, followed by a conclusion in Section 5. 
\subsection{Motivation}
\begin{itemize}
	\item It is known that deterministic finite automata equipped with queue can perform universal computations. In fact, it has been proved several times that the computational power of deterministic queue automata and Turing machine is equivalent.
	
	\item QTM is more efficient than classical Turing machine from a computational complexity point of view. For example, integer factorization and discrete logarithm problems are intractable on classical Turing machine, but they are tractable on QTM \cite{17}. Moreover, real-time QTM is more powerful than real-time DTM \cite{31}.
	\item	Motivated by the efficiency of QTM and classical queue automata, we investigate a quantum version of classical queue automata and its superiority over classical counterparts in real-time.
\end{itemize}
\section{Preliminaries and definitions}
In this section, we review some formal definitions. Throughout the paper, we have used the following notations: the prefix \textquotedblleft R\textquotedblright, \textquotedblleft D\textquotedblright, \textquotedblleft ND\textquotedblright, \textquotedblleft Q\textquotedblright~and  \textquotedblleft rt\textquotedblright ~signify reversible, deterministic, non-deterministic, quantum and real-time respectively. An input alphabet $\Sigma$ does not contains left and right-end markers (\#, \$), empty queue symbol $\perp$ is not an element of queue alphabet $\Sigma_q$. The length of input string \textit{x} is denoted by $|x|$. We assume that the reader is familiar with the notation of quantum computation; otherwise, reader can refer to \cite{6, 34} for quantum models. 

\begin{defn} 
	A deterministic queue automaton (DQA) is defined as a septuple ${\it (Q,\Sigma, \Gamma, \delta, q_0, \perp, F)}$, where 
	\begin{itemize}
		\item {\it Q} is a set of states, 
		\item {\it $\Sigma$} is an input alphabet,
		\item $\Gamma$ is a finite set of queue symbols,
		\item {\it $q_0$} is a starting state,
		\item $\perp$ is an empty queue symbol $\perp \neq \Gamma$,
		\item The transition function ${\it \delta}$ is defined by ${\it Q\times\Sigma \cup \{\lambda\}\times (\Gamma \times \Gamma) \cup (\{\perp\}\times \{\perp\}) \rightarrow Q \times \Gamma \cup \{\lambda\} \times \{\text{keep, remove}\}}$, where $\lambda$ signifies empty symbol. It must never be used as an input symbol. 
		\item ${\it F}$ is a set of accepting states ($F \subseteq Q$).
	\end{itemize}
	In DQA, it is possible to enqueue the symbol at the rear of queue and dequeue (remove) or keep the symbol at the front of the queue. In order to process the input string \textit{x} by $M_{DQA}= (Q,\Sigma, \Gamma, \delta,$ $q_0, \perp, F)$, we assume that \textit{x} is written on input tape employed with a queue. A computation process of $M_{DQA}$ is a sequence of configurations $c_0, c_1, c_2,...$ , where $c_0$ is an initial configuration. The configuration of $M_{DQA}$ is defined as a quadruple $(p, q , r, s)$, where $p \in \Sigma^*$ denotes the already read part of input string, $q \in Q$ is the present state, $r \in \Sigma^*$ is unread part of \textit{x} and \textit{s} indicates the content of queue, where leftmost symbol is at the front of queue. Suppose the configuration of $M_{DQA}$ is $c_1=(p, q_1, r_1r_2, z_1z_2...z_l)$, where $M_{DQA}$ is in state $q_1$ and R/W head is under the symbol $r_1$ and $z_1, z_l$ are the symbols at the front and read of the queue respectively, where $r_1 \in \Sigma \cup \{\lambda\}, p, r_2 \in \Sigma^{*}, z^{'} \in \Gamma \cup \{\lambda\}$  and $z_1,z_2,...z_l \in \Gamma ~\text{for}~ l \geq 1$. Thus, after reading the symbol $r_1$, the resultant configuration $c_2$ is as $c_2=(pr_1, q_2, r_2, z_1z_2...z_lz^{'})$. Thus, the above configuration $c_1$ is transformed into $c_2$ if the transition function is $\delta(q_1, r_1, z_1, z_l)= (q_2, z^{'}, \text{keep})$. Similarly, in order to remove the symbol from head and put the other symbol at rear, then the resultant configuration will be $c_2=(pr_1, q_2, r_2, z_2...z_lz^{'})$ if transition function is $\delta(q_2, r_3, z_1, z_l)= (q_3, z^{'}, \text{remove})$. In case, the queue is empty at initial stage, then the resultant configuration is computed by $\delta$ with empty queue symbol such as $\delta(q_1, r_1, \perp, \perp)= (q_2, z^{'}, \text{keep})$. Fig 1 shows the pictorial representation of above configurations of a DQA.
\end{defn}

\begin{figure*}
        \centering
                \includegraphics[scale=0.5]{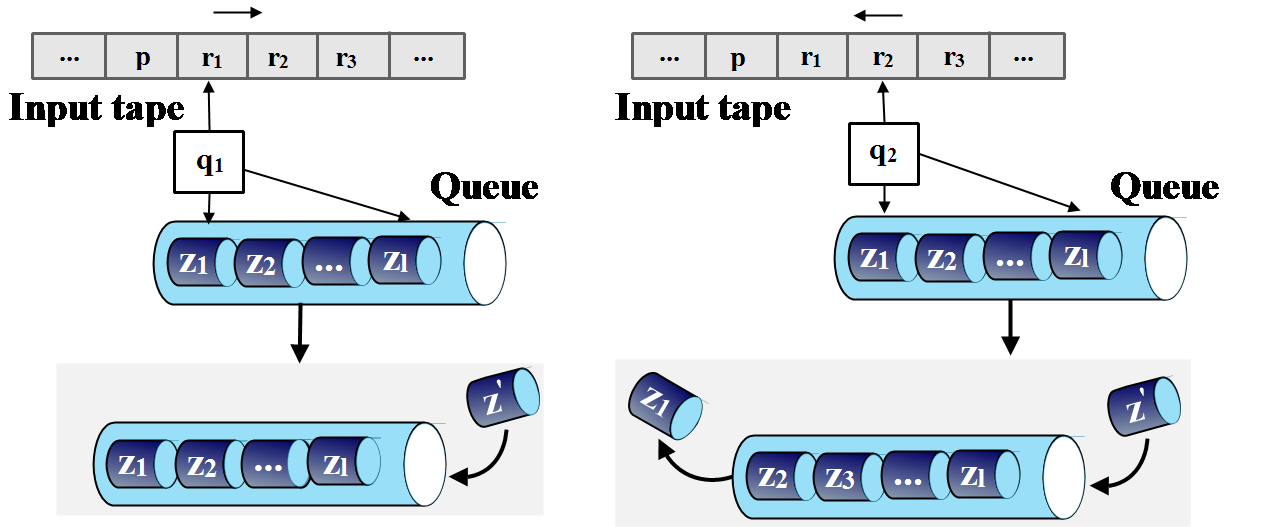}
                	\caption{Resultant configurations of a DQA} \label{1}
\end{figure*}

The definition of reversible DQA (RDQA) is same as above DQA. The only difference lies in the transitions i.e. any configuration of RDQA must have one configuration which can be computed by DQA. It has also to be deterministic backward and the symbols are enqueued at the front and dequeued from rear of the queue. Any automaton is said to be real-time if it complete its computation in real-time i.e the running time of automaton on any input string $x \in \Sigma^{*}$ is less than equal to $|x|$. But, the notion real-time depends upon the model of automata studied. Kutrib et al. \cite{25} studied the restricted versions of DQA by putting restriction on time. A RDQA is said to be real-time DQA (rtRDQA) if there are no $\lambda$-steps in computation. It is said to be quasi real-time if there is a constant number of $\lambda$-steps applicable for all computations.

\section{Quantum queue automata}
Quantum queue automata (QQA) is constructed similarly as quantum pushdown automata (QPDA). QQA employs a data structure \textit{queue} referred to as First-In, First-Out (FIFO), whereas QPDA employs stack which referred to as Last-In, First-Out (LIFO) \cite{11}. QQA consists  input tape, queue and finite state control. QQA can be defined as a modification of classical queue automata by adding weighted superposition to the configurations of classical queue automata in such a way that processing of input string corresponds a unitary transformation. QQA choose a transition by considering current state and the symbols at the front and end of the queue. 

\begin{defn}
	A quantum queue automaton $M_Q$ is defined as septuple $(Q, \Sigma, \Sigma_q, q_0, Q_{acc},$ $Q_{rej}, \delta)$, where
	\begin{itemize}
		\item $\it Q$ is a set of states. Moreover, $Q=Q_{acc}\cup Q_{rej}\cup Q_{non}$, where $Q_{acc}, Q_{rej}, Q_{non}$ represent the set of accepting, rejecting and non-halting states respectively. 
		\item $\Sigma$ is an input alphabet,
		\item $\Sigma_q$ is a queue alphabet such that $\Sigma_q= \Sigma \cup \{\#,\$\}$, where end-markers \{\#,\$\} are not included in $\Sigma$. For convenience, we have used $\Sigma_{\tau}~ \text{for}~ \Sigma_q \cup \{\tau\}$ such that $\tau$ represents an empty word and $\perp$ is the empty queue symbol such that ($\perp \notin \Sigma_q$),
		\item Transition function $\delta$ is defined by $\delta: Q\times\Sigma\times\Sigma_q \times \Sigma_q \times Q\times \Sigma_{\tau} \times D \times X\rightarrow C$, where $D=\{\leftarrow,\uparrow,\rightarrow\}$ represent the head function for the left, stationary and right direction of R/W head, $X \in \{\epsilon, \omega\}$ where $\epsilon, \omega$ represent the dequeue and enqueue operations respectively. In following conditions, we have used $z_1, z_2,...z_l \in \Sigma_q ~ \text{for}~ l \geq 1$, where $z_1, z_l$ signifies first and last symbol of queue respectively  Transition function must satisfy the following conditions:
	
	\begin{enumerate} [(a)]
	
			\item Local probability condition:
			\begin{equation}	
			\smashoperator[r]{\sum\limits_{(q \textasciiacute, z \char`\', d, \omega)\in Q\times \Sigma_{\tau} \times D \times X }^{\forall (q_1, \sigma, z_1, z_l) \in Q \times \Sigma \times \Sigma_q \times \Sigma_q}}  {\mid \delta(q_1, \sigma, z_1, z_l, q \char`\', z \char`\', d, \omega)\mid^2}=1
			\end{equation}
			\item Orthogonality condition: 
			\begin{equation}
			\smashoperator[r]{\sum_{(q\char`\', z\char`\', d, \omega) \in ~ Q\times \Sigma_{\tau} \times D \times X}^{\forall (q_1, \sigma, z_1, z_l) \neq (q_2, \sigma, z_1, z_l)~ \text{in} ~ Q \times \Sigma \times \Sigma_q \times \Sigma_q }} {\overline{\delta(q_1, \sigma, z_1, z_l, q\char`\', z\char`\', d, \omega)}~ \delta(q_2, \sigma, z_1, z_l, q\char`\', z\char`\', d, \omega)}=0
			\end{equation}
			\item Separability condition I: 
			\begin{itemize}
				\item	
				\begin{equation}
				\begin{split}
				& \smashoperator{\sum\limits_{\substack{(q\char`\', z\char`\', d, \epsilon)\in Q\times \Sigma_{\tau} \times D \times X, \\ (q\char`\', z\char`\', d, \omega)\in Q\times \Sigma_{\tau} \times D \times X}}^{\forall (q_1, \sigma_1, z_1, z_l), (q_2, \sigma_1, z_1, z_l) \in Q \times \Sigma \times \Sigma_q \times \Sigma_q}} {\overline{\delta(q_1, \sigma_1, z_1, z_l, q\char`\', z\char`\', d, \epsilon)}~ \delta(q_2, \sigma_1, z_1, z_l, q\char`\', z\char`\', d, \omega)} \\
				& + \smashoperator{\sum\limits_{(q\char`\', z\char`\', d, \epsilon)\in Q\times \Sigma_{\tau} \times D \times X}}{\overline{\delta(q_1, \sigma_1, z_1, z_l, q\char`\', z\char`\', d, \epsilon)}~ \delta(q_2, \sigma_1, z_1, z_l, q\char`\', z\char`\', d, \epsilon)}=0
				\end{split}
				\end{equation}
				\item  
				\begin{equation}
				\smashoperator{\sum\limits_{\substack{(q\char`\', z\char`\', d, \epsilon)\in Q\times \Sigma_{\tau} \times D \times X, \\ (q\char`\', z\char`\', d, \omega)\in Q\times \Sigma_{\tau} \times D \times X}}}{\overline{\delta(q_1, \sigma_1, z_1, z_l, q\char`\', z\char`\', d, \epsilon)}~ \delta(q_2, \sigma_1, z_1, z_l, q\char`\', z\char`\', d, \omega)}=0
				\end{equation}
			\end{itemize}	
			\item Separability condition II:
			\begin{equation}
			\smashoperator{\sum\limits_{(q\char`\', z\char`\', \omega)\in Q\times \Sigma_{\tau} \times X}^{\forall (q_1, \sigma_1, z_1, z_l), (q_2, \sigma_2, z_1, z_l) \in Q \times \Sigma \times \Sigma_q \times \Sigma_q}}{\overline{\delta(q_1, \sigma_1, z_1, z_l, q\char`\', z\char`\', \rightarrow, \omega)}~ \delta(q_2, \sigma_2, z_1, z_l, q\char`\', z\char`\', \uparrow, \omega)}=0
			\end{equation}
			\item Separability condition III: 
			\begin{itemize}
				\item	
				\begin{equation}
				\smashoperator{\sum\limits_{\substack{(q\char`\', z\char`\', \epsilon)\in Q\times \Sigma_{\tau} \times X, \\ (q\char`\', z\char`\',  \omega)\in Q\times \Sigma_{\tau} \times X}}^{\forall (q_1, \sigma_1, z_1, z_l), (q_2, \sigma_2, z_1, z_l) \in Q \times \Sigma \times \Sigma_q \times \Sigma_q, \forall d_1, d_2 \in D: d_1 \neq d_2}}{\overline{\delta(q_1, \sigma_1, z_1, z_l, q\char`\', z\char`\', d_1, \epsilon)}~ \delta(q_2, \sigma_2, z_1, z_l, q\char`\', z\char`\', d_2, \omega)}=0 
				\end{equation}
				\item  \begin{equation}
				\smashoperator{\sum\limits_{\substack{(q\char`\', z\char`\', \omega)\in Q\times \Sigma_{\tau} \times X, \\ (q\char`\', z\char`\', \epsilon)\in Q\times \Sigma_{\tau} \times X}}}{\overline{\delta(q_1, \sigma_1, z_1, z_l, q\char`\', z\char`\',  d_1, \omega)}~ \delta(q_2, \sigma_2, z_1, z_l, q\char`\', z\char`\', d_2, \epsilon)}=0 
				\end{equation}
				\item \begin{equation}
				\smashoperator{\sum\limits_{\substack{(q\char`\', z\char`\', d, \epsilon)\in Q\times \Sigma_{\tau} \times D \times X, \\ (q\char`\', z\char`\', d, \omega)\in Q\times \Sigma_{\tau} \times D \times X}}^{\forall (q_1, \sigma_1, z_1, z_l), (q_2, \sigma_2, z_1, z_l) \in Q \times \Sigma \times \Sigma_q \times \Sigma_q}}{\overline{\delta(q_1, \sigma_1, z_1, z_l, q\char`\', z\char`\', d, \epsilon)}~ \delta(q_2, \sigma_2, z_1, z_l, q\char`\', z\char`\', d, \omega)}=0
				\end{equation}
			\end{itemize}	
		\end{enumerate}
	\end{itemize}
\end{defn}

To process the input string by $M_Q$, we assume that input string \textit{x} is written on input tape enclosed with both end-markers such as \textit{\#x\$}. It processes the input tape employed with queue which is potentially infinite on the right-side. The automaton is  in the state \textit{q}, R/W head is above the symbol $\sigma$. Then, $M_Q$ with the amplitude $\delta(q, \sigma, z_1, z_l, q\char`\', z\char`\', d, \omega)$, where $z_1$, $z_l$ are the symbols at the front and end of queue respectively. It moves to state $q\char`\'$, $d\in \{\leftarrow, \uparrow, \rightarrow\}$ moves the R/W head one cell towards left, stationary and in right direction, and puts the symbol $z\char`\'$ at the end of the queue. The automaton $M_Q$ for processing an input \textit{x} corresponds a unitary evolution in the inner-product space $H_n$. Definition 1 utilizes the concept of deterministic queue automata \cite{14, 25} and quantum pushdown automata \cite{9, 11}.

A computation of QQA $M_Q$ is a sequence of superpositions $c_0, c_1, c_2, ...,$  where $c_0$  is an initial configuration. When the automaton is observed in a superposition state, for any $c_i$, it has the form $U_\delta \ket{c_i} = \sum_{c \in C_n}{\alpha_c \ket{c_i}}$, where \textit{C} defines the set of configurations, and the configuration $c_i$ is measured with the probability $\alpha_c$ \cite{5}. Superposition is valid; if the sum of the absolute squares of their amplitudes is unitary.

Time evolution of quantum systems is given by unitary transformations. Suppose if the system is in $\ket{\psi}$, then at a later time it will be $\ket{\psi^{'}}=\hat{U} \ket{\psi}$, where $\hat{U}$ is unitary time evolution operator. If \textit{} is any linear transformation, then it will be unitary transformation if $\overline{U}U=I$ or $U\overline{U}=I$, where $\overline{U}$ is a conjugate transpose of \textit{U}. Therefore, evolution operator specifies how QQA $M_Q$ will progress the input string. Each transition function $\delta$ induces a linear time evolution operator over the space $H_n$. Let $\sigma_1, \sigma_2, ..., \sigma_n \in \Sigma ~\text{and}~ z_1, z_2, ..., z_l  \in \Sigma_q ~\text{for}~ l \geq 1$. For any configuration $c=\ket{\sigma_1q\sigma_2\sigma_3...\sigma_n, z_1 z_2...z_l}$, the evolution of a QQA $M_Q$ is given by the linear operator $U_\delta$ such that
\begin{equation}
U_\delta \ket{c}= \sum\limits_{(q\char`\', z\char`\', d, \omega)\in Q\times \Sigma_{\tau} \times D \times X}{\delta(q, \sigma_2, z_1, z_l, q\char`\', z\char`\', d, \omega) \ket{f (\ket{c}, d, q\char`\'), z\char`\', \omega}}
\end{equation}
where $(q, \sigma, z_1, z_l) \in Q \times \Sigma \times \Sigma_q \times \Sigma_q$ and 
\begin{equation}
f(\ket{\sigma_1q\sigma_2...\sigma_n, z_1z_2...z_l z\char`\'}, d, q\char`\')= 	\left\{
\begin{array}{ll}
q\char`\'\sigma_2,z_1z_2...z_l z\char`\'  ~ \text{if} ~ d= \uparrow \\
q\char`\'\sigma_3, z_1z_2...z_l z\char`\'  ~ \text{if} ~ d= \rightarrow \\
q^{'}\sigma_1, z_1z_2...z_l z\char`\'   ~  \text{if} ~ d= \leftarrow \\
\end{array} \right \} 
\end{equation}

When R/W head is stationary in equation (10), then the resultant state is reading the same symbol on input tape and enqueue the symbol $z\char`\'$ at the end of queue. The resultant state $q\char`\'$ reads the next symbol on moving the R/W head towards right. Further, the resultant state reads the $\sigma_1$ on moving towards the left direction and puts the symbol $z\char`\'$ at the rear of the queue. 
\begin{thm}
	Well-formedness conditions 1.1 are satisfied iff the time evolution operator $U_\delta$ is unitary. 
\end{thm}
\begin{proof}
	For each input \textit{x}, $U_\delta$ is unitary iff the vectors $U_\delta \ket{c}$ for $c \in C_x$ of the QQA evolution matrix are orthonormal. Condition (a) i.e. local probability condition is satisfied to the statement that $\lVert U_\delta \ket{c}\rVert=1$ for each $c \in C_x$ configuration Correspondingly, the column vectors of the QQA evolution matrix are orthogonal iff Condition (b) i.e. orthogonality condition is satisfied such that $U_\delta \ket{c_1}-U_\delta \ket{c_2}$, where $c_1=\ket{q_1 \sigma, z_1 z_2...z_l}, c_2=\ket{q_2 \sigma, z_1 z_2...z_l}$ for $q_1 \neq q_2$. Separability conditions (c, d, e) are satisfied in equivalent to the above statement. Separability condition d is satisfied such that $U_\delta \ket{c_1}-U_\delta \ket{c_2}$, in which states $q_1$ and $q_2$ are reading the different symbols and resultant state is same, where $\{c_1=\ket{q_1 \sigma_1, z_1 z_2...z_l}, c_2=\ket{q_2 \sigma_2, z_1 z_2...z_l}\mid c_1, c_2 \in C_x\}$, where R/W head moves right and remains stationary for $c_1, c_2$ respectively according to the condition. Thus, the columns of the evolution matrix are orthogonal for each input \textit{x} iff conditions (b, c, d, e) are satisfied. We can say that, if $U_\delta$ is a unitary operator, transition function $\delta$ satisfies the well-formedness conditions 1.1.

	It is not an easy task to check all the well-formedness conditions for trivial QQA (i.e. there is no other state that it can exist in). Consider a QQA, whose evolution matrix columns are orthonormal for each configuration, but the evolution is not unitary. $Q=\{q\}, \Sigma=\{a\}, \Sigma_q=\{A_0\}$ and transition function $\delta$ is defined as: $\delta(q, \#, A_0, A_0, q, A_0, \rightarrow, \omega), \delta(q, a, A_0, A_0, q, A_0, \rightarrow, \omega)=1, \delta(q, \$, A_0, A_0, q, \$, \rightarrow, \omega)=1$. The other values of the transitions $\delta=0$. Therefore, we have defined simple notation of QQA similar as 2QFA and QPA, by which well-formed machines can be more easily specified. The method is to decompose the transition function into transforming of states with queue automata operations (enqueue, dequeue) and other head functions.
\end{proof}

\begin{defn}
	A QQA is simplified, for each $\sigma \in \Sigma$, if there exists a function $D: Q \rightarrow \{\leftarrow, \uparrow, \rightarrow\}$ on the inner product space $L_2(Q) \rightarrow L_2(Q)$ such that  where \textit{Q} is the set of states, $X \in \{\epsilon, \omega\}$. Define transition function as	
	\begin{equation}
	\left\{ \begin{array}{l}
	\varphi(q, \sigma, z_1, z_l, q\char`\', z\char`\', \epsilon)= \delta(q, \sigma, z_1, z_l, q\char`\', z\char`\', D(q\char`\'), \epsilon)
	\\
	0 \end{array}
	\middle\vert\;
	\begin{array}{@{}l@{}}
	\text{if} ~ D(q\char`\')=d \\
	\text{else}
	\end{array}
	\right\}
	\end{equation}
	where a state \textit{q} results in to $q\char`\'$ on reading $\sigma$, dequeue a symbol $z_1$ from head and puts the $z^{'}$ at end of the queue.
\end{defn}
\begin{thm}
	A simple QQA satisfies the well-formedness conditions 1.1 if there exists a transition function $\varphi(q, \sigma, z_1, z_2, q\char`\', z\char`\', \omega)$ for any $\sigma \in \Sigma$, $D: Q \rightarrow \{\leftarrow, \uparrow, \rightarrow\}$ on the inner product space $L_2(Q) \rightarrow L_2(Q)$ and $X \in \{\epsilon, \omega\}$ such that 
	\begin{equation}
	\smashoperator[r]{\sum_{(q\char`\', z\char`\', \omega)\in  Q\times \Sigma_{\tau} \times X}^{\forall (q_1, \sigma, z_1, z_l), (q_2, \sigma, z_1, z_l)~ \in ~ Q \times \Sigma \times \Sigma_q \times \Sigma_q }} {\overline{\varphi(q_1, \sigma, z_1, z_l, q\char`\', z\char`\', \omega)}~ \varphi(q_2, \sigma, z_1, z_l, q\char`\', z\char`\', \omega)}= \left\{ \begin{array}{l}
1 \\
0 \end{array}
\middle\vert\;
\begin{array}{@{}l@{}}
q_1=q_2 \\
q_1 \neq q_2
\end{array}
\right\}
\end{equation}
\end{thm}

\begin{proof}
Firstly re-write the well-formedness conditions:
\begin{equation*}
\begin{split}
& \sum\limits_{(q\char`\', z\char`\', d, \omega)\in Q\times \Sigma_{\tau} \times D \times X }{\overline{\delta(q_1, \sigma, z_1, z_l, q\char`\', z\char`\', d, \omega)}~ \delta(q_2, \sigma, z_1, z_l, q\char`\', z\char`\', d, \omega)}=  \\
&  \sum\limits_{(q\char`\', z\char`\', d, \omega)\in Q\times \Sigma_{\tau} \times D \times X }{\overline{\delta(q_1, \sigma, z_1, z_l, q\char`\', z\char`\', D(q\char`\'), \omega)}~ \delta(q_2, \sigma, z_1, z_l, q\char`\', z\char`\', D(q\char`\'), \omega)}= \\
\end{split}
\end{equation*}
\begin{equation}
\smashoperator[r]{\sum_{(q\char`\', z\char`\', \omega)\in  Q\times \Sigma_{\tau} \times X}^{\forall (q_1, \sigma, z_1, z_l) , (q_2, \sigma, z_1, z_l)~ \in ~ Q \times \Sigma \times \Sigma_q \times \Sigma_q }} {\overline{\varphi(q_1, \sigma, z_1, z_l, q\char`\', z\char`\', \omega)}~ \varphi(q_2, \sigma, z_1, z_l, q\char`\', z\char`\', \omega)}= \left\{ \begin{array}{l}
1 \\
0 \end{array}
\middle\vert\;
\begin{array}{@{}l@{}}
q_1=q_2 \\
q_1 \neq q_2
\end{array}
\right\}
\end{equation}
We can see that $M_Q$ is well-formed iff the above condition is satisfied, for every $\sigma \in \Sigma$ and $X \in \{\epsilon, \omega\}$. The local probability condition (a) is satisfied iff the columns of every transition $\sigma$ are normalized (i.e. length 1) for each enqueue and dequeue operation. Similarly, the columns vectors of transition are orthogonal iff the orthogonality condition (b) and separability conditions (c, d, e) are satisfied. Equivalently, $M_Q$  is well-formed when every transition is unitary for $\sigma \in \Sigma$ and $X \in \{\epsilon, \omega\}$.
\end{proof}

The definition of a real-time quantum queue automaton (rtQQA) is same as QQA. The only difference lies in the movement of R/W head. It is only allowed to move towards the right direction Thus, every step rtQQA reads a new symbol. On reading the right-end marker \$, the computation is finished. The input string is said to be recognized by rtQQA if the R/W reads the right-end marker \$ and queue is empty. Thus, rtQQA accepts the language \textit{L}, if  it takes time not more than $|w|$ for every word $w \in L$.

\subsection{Language recognition}
We assume that QQA has to be observed to produce information about its processing. Consider an observable \textit{O} for finite-dimensional Hilbert space $H_n$, which is decomposed into subspaces such as $E_a, E_r, E_n$  refers to the subspace of ‘accept’, ‘reject’ and ‘non-halting’ respectively. Each of these subspaces are traversed by configurations such that $c_a=\{\ket{q \sigma_1 \sigma_2...\sigma_n, z_1z_2...z_l}\in C \mid q \in Q_{acc}\}, c_r=\{\ket{q \sigma_1 \sigma_2...\sigma_n, z_1z_2...z_l}\in C \mid q \in Q_{rej}\}$ and $c_n=\dfrac{C}{(c_a \cup c_r)}$, where $c_n \in C$.

We assume that input string $x \in \Sigma^{*}$ is written on input tape with both end-markers such that \#x\$. It is equipped with a queue (i.e. initially empty). The processing of input string starts with an initial state and R/W points towards the first symbol on input tape. The transition depends upon the symbol under R/W head and the symbols at the front and end of the queue respectively. Firstly, the evolution operator is applied to the current state and computes several paths simultaneously (quantum parallelism); however, as a result of measurement, it is possible to get the results of only one computation path. In meanwhile, each path performs an enqueue or dequeue operations on queue and moves the R/W head corresponding to the resultant state. Then the result is observed by an observable \textit{O}.  Suppose if the automaton is in a superposition state $\ket{\phi}=\alpha_1\ket{x_1}+\alpha_2\ket{x_2}+...+\alpha_n\ket{x_n}$, where $\alpha_i$ are amplitudes and $|\alpha_1|^2+ |\alpha_2|^2+...+|\alpha_n|^2=1$, then the superposition is projected into above-mentioned subspaces $E_j, j\in \{a, r, n\}$. The result is observed randomly, and each result \textit{j} is realized with probability $\lVert \alpha_j\rVert^2$. The result of each observation will be either ‘accept’ or ‘reject’ or ‘non-halting’. The processing remains continue until the observation does not undergo ‘acceptance’ or ‘rejectance’ state. Therefore, when the R/W head reaches at right end of the input tape and queue is empty, then the string is said to accepted otherwise the string is rejected after processing the input string. 

\section{The power of quantum queue automata}
The computational power of automata employed with queue has been widely studied by various researchers. It is known that queue automata and Turing machines have the same computational power  i.e capable of performing universal computations. Kutrib et al. \cite{25} shown the several languages recognized by rtRDQA and proved its closure properties. It has been examined that languages $L_1=\{ba^nca^nb \mid n \geq 0\}~ \text{and} ~ L_2=\{ba^nba^mca^mb \mid m,n \geq 0\}$ can be recognized by some reversible DQA in real-time. But, the union of $L_3=L_1 \cup L_2$ i.e $L_3=\{ba^{n_1}ba^{n_2}ba^{n_3}...ba^{n_i}ca^{n_i}b \mid n_j \geq 0, 1 \leq j \leq i\}$ cannot be recognized by any rtDQA \cite{25}. Further, it has been investigated that $L_{xy}=\{xycyx \mid x \in \{a, b\}^*, y \in \{0, 1\}^*\}$ cannot be recognized by any non-deterministic queue automata in real-time (rtNDQA) \cite{26}. In this Section, we have investigated the computational power of real-time quantum queue automata. Thus, by Theorems 4.1 and 4.2, the real-time quantum queue automata has been shown to outperform its classical variants in the regime of language recognition by imposing same restrictions.

\begin{thm}
A language $L_3=\{ba^{n_1}ba^{n_2}ba^{n_3}...ba^{n_i}ca^{n_i}b \mid n_j \geq 0, 1 \leq j \leq i\}$ can be recognized by real-time quantum queue automata, but cannot be recognized by any deterministic queue automata in real-time.
\end{thm}
\begin{proof}
Let $M_{rtQQA}=(Q, \Sigma, \Sigma_q, q_0, Q_{acc},$ $Q_{rej}, \delta)$ be a real-time QQA, $Q=\{q_0, q_1, q_2, q_3, q_4, q_5,$ $Q_{acc}, Q_{rej}\}, \Sigma =\{a, b, c\}, \Sigma_q=\{A, B\}, Q_{acc}=\{q_{acc_1}, q_{acc_2}\}, Q_{rej}= \{q_{rej_1}\}$. The transition function $\delta$ is defined in the manner as described in Section 3. It must be noted that the head moves always towards the right direction on reading a new symbol at each step. The specification of transition functions is defined as follows:

\begin{table} [h]
\centering
\caption{List of transition functions for language $L_3=\{ba^{n_1}ba^{n_2}b...ba^{n_j}ca^{n_j}b\mid n_k\geq 0, 1 \leq k \leq j\}$}. 
\begin{tabular}{ |p{5.2cm}|p{5.2cm}|}
	\hline
	\multicolumn{2}{|c|}{$\varphi(q_0, \#, \perp, \perp)=(q_0, \tau, \omega)$}\\
	\multicolumn{2}{|c|}{$\varphi(q_0, b, \perp, \perp)=\dfrac{1}{\sqrt{2}}(q_1, \tau, \omega)+ \dfrac{1}{\sqrt{2}}(q_2, \tau, \omega)$}\\
	\hline
	$\varphi(q_1, a, \perp, \perp)=(q_1, A, \omega)$ & $\varphi(q_2, a, \perp, \perp)=(q_2, \tau, \omega) $ \\
	\hline
	$\varphi(q_1, a, A, A)=(q_1, A, \omega)$ & $\varphi(q_2, b, \perp, \perp)=(q_1, \tau, \omega) $ \\
	\hline
	$\varphi(q_1, b, A, A)=(q_r, \tau, \omega)$ & $\varphi(q_2, b, \perp, \perp)=(q_2, \tau, \omega) $ \\
	\hline
	$\varphi(q_1, c, A, A)=(q_3, \tau, \omega)$ & $\varphi(q_2, c, \perp, \perp)=(q_4, \tau, \omega) $ \\
	\hline
	$\varphi(q_3, a, A, A)=(q_3, \tau, \epsilon)$ & $\varphi(q_4, a, \perp, \perp)=(q_4, \tau, \omega) $ \\
	\hline
	$\varphi(q_3, a, \perp, \perp)=(q_{rej_1}, \tau, \epsilon)$ & $\varphi(q_4, b, \perp, \perp)=(q_4, \tau, \omega) $ \\
	\hline
	$\varphi(q_3, b, A, A)=(q_{rej_1}, \tau, \epsilon)$ & $\varphi(q_3, b, \perp, \perp)=(q_5, \tau, \omega) $ \\
	\hline
	$\varphi(q_5, \$, \perp, \perp)=(q_{acc_1}, \tau, \omega)$ & $\varphi(q_4, \$, \perp, \perp)=(q_{acc_2}, \tau, \omega)$  \\
	\hline
\end{tabular}
\end{table}
In Table 1, transition functions are applicable in the case where $M_{rtQQA}$ is in state $q \in Q$
and R/W is above the symbol $\sigma \in \Sigma$ and  $z_1$, $z_l$ are the symbols at the front and end of queue respectively are represented as:
\begin{equation}
\varphi(q, \sigma , z_1, z_l)=\sum\limits_{(q\char`\', z\char`\', \omega)\in Q\times \Sigma_{\tau} \times X }{\varphi(q, \sigma, z_1, z_l, q\char`\', z\char`\', \omega)(q\char`\', z\char`\', \omega)}
\end{equation} 

$M_{rtQQA}$ starts by splitting into two computational paths. Each path possesses equal amplitude $1/\sqrt{2}$. In the first path, state $q_1$ reads \textit{a} and empty queue symbol $\perp$, then it enqueues symbol \textit{A} at the rear of queue and moves the R/W towards right. This process continues and state remains same. On reading symbol \textit{b}, the $M_{rtQQA}$ change its state to rejectance state $q_r$. The state $q_1$ is changed into $q_3$ on reading a symbol \textit{c} from input tape and symbol \textit{A} is at the front and end of queue. Further, it starts dequeue each \textit{A} from front of queue on reading \textit{a} from tape. If in case, state $q_3$ reads \textit{a} and queue gets empty or reads \textit{b} symbol and symbol \textit{A} is at the front and end of queue, then it goes to rejectance state $q_{rej_1}$. Otherwise, the state $q_3$ is changed in to $q_5$ on reading \textit{b} and empty queue symbol $\perp$. \\

While in the second path, R/W keeps moving towards right of the input tape and neither enqueue nor dequeue any symbol from the queue. Whenever a state $q_2$ reads a symbol \textit{b}, it splits the computation into two paths, where first path follows the above procedure and other path follows the loop. But, on reading state a symbol \textit{c} and empty queue, state $q_2$ is changed into state $q_4$ and moves to the right. If the first path finds the difference in number of \textit{a}'s before the symbol \textit{c} and after it, then it goes to rejectance state. Otherwise, on reading symbol \$ from the input tape and empty queue symbol, the working states $q_5$ and $q_4$ go to the accepting states $q_{acc_1}$ and $q_{acc_2}$ respectively. At the end, the total amplitude can be written as the product of the amplitudes associated with each subpath. Thus, the number of subpaths depends upon the number of \textit{b}'s occur before the symbol \textit{c} in second path.  If the input string \textit{w} $\in L_3$, then both the paths reads the right-end marker \$ i.e the string is said to be accepted with probability at most 1. If \textit{w} $\notin L_3$, then it is rejected by one of the path and the input string is said to be rejected with probability at most 1/2. Thus, the inputs which are not in $L_1$ are accepted with probability at most $1/2^i$, where \textit{i} depends upon the number of \textit{b}'s occur before the symbol \textit{c}. For instance, consider an input string written on input tape as \textit{\#bcb\$} enclosed with both end markers. On reading the first \textit{b}, the computation is split into two paths with $\dfrac{1}{\sqrt{2}}$. In both paths, state $q_1$ and state $q_2$ reads the next symbol \textit{c} with an empty queue and changed into state $q_3$ and $q_4$ respectively. Finally, the both paths go to the accepting states and string is said to be accepted with probability 1. If the input string is taken as \textit{w}=\textit{\#bacb\$} i.e. $w \notin L_3$, then it is rejected with probability 1/2. 
\end{proof}

\begin{thm}
There exists a language $L_{xy}=\{xycyx \mid x \in \{a, b\}^*, y \in \{0, 1\}^*\}$, that can be recognized by real-time quantum queue automata, but cannot be recognized by non-deterministic queue automata in real-time.
\end{thm}
\begin{proof}
Let $M_{rtQQA}=(Q, \Sigma, \Sigma_q, q_0, Q_{acc},$ $Q_{rej}, \delta)$ be a real-time QQA, $Q=\{q_0, q_1, q_2, q_3, q_4,$ $ Q_{acc}, Q_{rej}\}, \Sigma =\{a, b, 0, 1\}, \Sigma_q=\{A, B\}, Q_{acc}=\{q_{acc_1}, q_{acc_2}\}, Q_{rej}= \{q_{r}\}$. Each transition in Table 2 is unitary by inspection and the other transitions are impulsive so that the transformations are unitary. For convenience, we have used the symbols $z_1, z_l \in \{A, B\}$ to represent the symbols at the front and rear of the queue. The specification of transition functions is defined as follows:

\begin{table} [h]
\centering
\caption{List of transition functions for language $L_{xy}=\{xycyx \mid x \in \{a, b\}^*, y \in \{0, 1\}^*\}$}
\begin{tabular}{ |p{5.3cm}|p{5.3cm}|}
	\hline
	\multicolumn{2}{|c|}{$\varphi(q_0, \#, \perp, \perp)=\dfrac{1}{\sqrt{3}}(q_1, \tau, \omega)+\dfrac{1}{\sqrt{3}}(q_2, \tau, \omega)+\dfrac{1}{\sqrt{3}}(q_r, \tau, \omega) $} \\
	\hline
	$\varphi(q_1, a, \perp, \perp)=(q_1, A, \omega)$ & $\varphi(q_2, a, \perp, \perp)=(q_2, \tau, \omega)$ \\
	\hline
	$\varphi(q_1, b, \perp, \perp)=(q_1, B, \omega)$ & $\varphi(q_2, b, \perp, \perp)=(q_2, \tau, \omega)$  \\
	\hline
	$\varphi(q_1, a, z_1, z_l)=(q_1, A, \omega)$ & $\varphi(q_2, 0, \perp, \perp)=(q_2, A, \omega)$ \\
	\hline
	$\varphi(q_1, b, z_1, z_l)=(q_1, B, \omega)$ &  $\varphi(q_2, 1, \perp, \perp)=(q_2, B, \omega)$  \\
	\hline
	$\varphi(q_1, 0, z_1, z_l)=( q_1, \tau, \omega)$ & $\varphi(q_2, 0, z_1, z_l)=( q_2, A, \omega)$ \\
	\hline
	$\varphi(q_1, 1, z_1, z_l)=( q_1, \tau, \omega)$ & $\varphi(q_2, 1, z_1, z_l)=( q_2, B, \omega)$ \\
	\hline
	$\varphi(q_1, c, z_1, z_l)=( q_3, \tau, \omega)$ & $\varphi(q_2, c, z_1, z_l)=( q_4, \tau, \omega)$ \\
	\hline
	$\varphi(q_3, 0, z_1, z_l)=( q_3, \tau, \omega)$ & $\varphi(q_4, 0, A, z_l)=( q_4, \tau, \epsilon)$ \\
	\hline
	$\varphi(q_3, 1, z_1, z_l)=( q_3, \tau, \omega)$ &  $\varphi(q_4, 1, B, z_l)=( q_4, \tau, \epsilon)$\\
	\hline
	$\varphi(q_3, a, A, z_l)=( q_3, \tau, \epsilon)$ & $\varphi(q_4, 0, B, z_l)=( q_r, \tau, \epsilon)$ \\
	\hline
	$\varphi(q_3, b, B, z_l)=( q_3, \tau, \epsilon)$ & $\varphi(q_4, 1, A, z_l)=( q_r, \tau, \epsilon)$ \\
	\hline
	$\varphi(q_3, a, B, z_l)=( q_r, \tau, \epsilon)$ & $\varphi(q_4, a, \perp, \perp)=( q_4, \tau, \epsilon)$  \\
	\hline
	$\varphi(q_3, b, A, z_l)=(q_r, \tau, \epsilon)$ & $\varphi(q_4, b, \perp, \perp)=(q_4, \tau, \epsilon)$ \\
	\hline
	$\varphi(q_3, \$, \perp, \perp)=(q_{acc_1}, \tau, \epsilon)$ &  $\varphi(q_4, \$, \perp, \perp)=( q_{acc_2}, \tau, \epsilon)$ \\
	\hline
\end{tabular}
\end{table}
Similarly to Table 1, transitions functions are applicable in the case where $M_{rtQQA}$ is in state $q \in Q$
and R/W is above the symbol $\sigma \in \Sigma$ and  $z_1$, $z_l$ are the symbols at the front and end of queue respectively are represented as:
\begin{equation}
\varphi(q, \sigma , z_1, z_l)=\sum\limits_{(q\char`\', z\char`\', \omega)\in Q\times \Sigma_{\tau} \times X }{\varphi(q, \sigma, z_1, z_l, q\char`\', z\char`\', \omega)(q\char`\', z\char`\', \omega)}
\end{equation}  

For the construction of $M_{rtQQA}$ for language $L_{xy}$, the computation process is split into three paths: one of which goes to the rejectance state $q_r$, and other two paths with states $q_1$ and $q_2$ compare the symbols representing \textit{x} and \textit{y} respectively. In the first path, state $q_1$ puts the symbol \textit{A} or \textit{B} into the queue on reading reads \textit{a} or \textit{b} with an empty queue symbol $\perp$ respectively and keep moving towards the right. On reading $y \in \{0, 1\}^*$, the state $q_1$ neither enqueue nor dequeue any symbol from queue and moves to the right. On reading \textit{c}, the state $q_1$ is changed into $q_3$. The content of the queue remains same on reading \textit{y}. On reading the symbols belong to \textit{x}, it compares the symbols and dequeue \textit{A} for each \textit{a} and symbol \textit{B} for each \textit{b} from the front of queue.  \\
In the second path, state $q_2$ enqueues the symbol \textit{A} or \textit{B} into the queue on reading reads \textit{0} or \textit{1} with an empty queue symbol $\perp$ respectively and keep moving towards the right. On reading $x \in \{a, b\}^*$, the state $q_2$ neither enqueue nor dequeue any symbol from queue and moves to the right. On reading \textit{c}, the state $q_2$ is changed into $q_4$. On reading the symbols belong to \textit{y}, it compares the symbols and dequeue \textit{A} for each \textit{0} and symbol \textit{B} for each \textit{1} from head of the queue. The content of the queue remains same on reading $x \in \{a,b\}^*$ at the end. \\
If any mismatch occurs in path on reading the symbols from input tape and queue, then it goes to the rejectance state. Otherwise, on reading the right-end marker \$, both paths goes to the acceptance states $q_{acc_1}$ and $q_{acc_2}$ respectively. Thus, R/W moves always towards the right direction on reading the symbol and $M_{rtQQA}$ finishes its work on reading the right-end marker \$. \\
If the input string \textit{w} $\in L_{xy}$, then both the paths reads the right-end marker \$ and goes to the accepting states i.e the string is said to be accepted with probability $\dfrac{2}{3}$. If the input string \textit{w} $\notin L_{xy}$, then it is rejected by at least one of the path and \textit{w} is said to be rejected with probability greater than equal to $\dfrac{2}{3}$.   
\end{proof}
\section{Conclusion}

In this paper, we have proposed a quantum variant QQA of classical queue automata. Further, we have introduced a generalized variant of real-time classical queue automata, the real-time quantum queue automata and proved that it is more superior than its classical counterparts. We have shown that there exists a language which can be recognized by real-time quantum queue automata, while it cannot be recognized by any real-time deterministic queue automata. Furthermore, we have shown that there is a language which cannot be recognized by real-time non-deterministic queue automata can be recognized by proposed quantum variant in real-time. Thus, we can conclude that the computational power of real-time queue automata has been increased by using quantum transitions.  

\section*{Acknowledgments}
Amandeep Singh Bhatia was supported by Maulana Azad National Fellowship (MANF), funded by Ministry of Minority Affairs, Government of India. 

\label{sec:test}
\bibliographystyle{elsarticle-num}
\bibliography{sample}
\end{document}